\documentclass[letterpaper, 10pt, conference]{ieeeconf}  

\usepackage{graphicx,wrapfig}
\usepackage[utf8]{inputenc}
\usepackage[T1]{fontenc}
\usepackage{mathtools}
\usepackage[english]{babel}
\usepackage{comment}
\usepackage{amsmath}
\usepackage{amssymb}
\usepackage{bm}
\usepackage{psfrag}
\usepackage{esdiff}
\usepackage{float}
\usepackage{caption}
\usepackage{subcaption}
\usepackage{bm}
\usepackage{algorithmicx}
\usepackage{algpseudocode}
\usepackage[ruled,vlined]{algorithm2e}
\usepackage{lipsum}
\usepackage{array}
\usepackage{blindtext}
\usepackage{epstopdf}
\usepackage{ragged2e}
\usepackage{mathrsfs}
\usepackage{tabu}
\graphicspath{ {figures/} }

\usepackage{nomencl}
\usepackage{cite}
\usepackage{lmodern}
\usepackage{enumerate}
\usepackage{tikz}
\usepackage{arcs}
\usepackage{etoolbox}
\usepackage{hyperref}

\IEEEoverridecommandlockouts                              

\DeclareMathOperator*{\argmin}{arg\,min}

\DeclarePairedDelimiter{\floor}{\lfloor}{\rfloor}

\newcommand{\be}{\begin{equation}}
\newcommand{\ee}{\end{equation}}
\newcommand{\bse}{\begin{subequations}}
	\newcommand{\ese}{\end{subequations}}
\newcommand{\bewn}{\begin{equation*}}
\newcommand{\eewn}{\end{equation*}}
\newcommand{\bbmat}{\begin{bmatrix}} 
	\newcommand{\ebmat}{\end{bmatrix}}
\newcommand{\bd}{\begin{displaymath}}
\newcommand{\ed}{\end{displaymath}}
\newcommand{\bea}{\begin{eqnarray}}
\newcommand{\eea}{\end{eqnarray}}
\newcommand{\ba}{\begin{array}}
	\newcommand{\ea}{\end{array}}
\newcommand{\baa}{\begin{array}{ll}}
	\newcommand{\eaa}{\end{array}}
\newcommand{\ds}{\displaystyle}
\newcommand{\bc}{\begin{center}}
	\newcommand{\ec}{\end{center}}
\newcommand{\ben}{\begin{enumerate}}
	\newcommand{\een}{\end{enumerate}}
\newcommand{\bi}{\begin{itemize}}
	\newcommand{\ei}{\end{itemize}}
\newcommand{\bt}{\begin{tabular}}
	\newcommand{\et}{\end{tabular}}
\newcommand{\bte}{\begin{table}}
	\newcommand{\ete}{\end{table}}

\newcommand{\norm}[1]{\left\lVert#1\right\rVert}   

\makeatletter
\renewcommand\paragraph{\@startsection{paragraph}{4}{\z@}%
	{-2.5ex\@plus -1ex \@minus -.25ex}%
	{1.25ex \@plus .25ex}%
	{\normalfont\normalsize\bfseries}}
\makeatother



\makeatletter
\providecommand\@gobblethree[3]{}
\patchcmd{\over@under@arc}
 {\@gobbletwo}
 {\@gobblethree}
 {}{}
\makeatother

\newtheorem{theorem}{Theorem}

\newtheorem{lemma}[theorem]{\textbf{Lemma}}

\newtheorem{problem}{\textbf{Problem}}
\newtheorem{definition}{\textbf{Definition}}

\newcommand{\bR}{\mathbb{R}}
\newcommand{\bZ}{\mathbb{Z}}
\newcommand{\calA}{\mathcal{A}}

\newcommand{\calD}{\mathcal{D}}
\newcommand{\calE}{\mathcal{E}}

\newcommand{\calG}{\mathcal{G}}

\newcommand{\calP}{\mathcal{P}}

\newcommand{\calS}{\mathcal{S}}

\newcommand{\calV}{\mathcal{V}}
\newcommand{\calW}{\mathcal{W}}
\newcommand{\calZ}{\mathcal{Z}}

\newcommand{\scriptA}{\mathscr{A}}
\newcommand{\scriptD}{\mathscr{D}}

\newcommand{\scriptF}{\mathscr{F}}
\newcommand{\scriptR}{\mathscr{R}}

\newcommand{\pdydx}[2]{\frac{\partial{#1}}{\partial{#2}}}

\newcommand{\abs}[1]{\left |#1\right |}

\title{\LARGE \bf
	Multi-Swarm Herding: Protecting against Adversarial Swarms
}

\author{Vishnu S. Chipade and Dimitra Panagou
	\thanks{The authors are with the Department of Aerospace Engineering,
		University of Michigan, Ann Arbor, MI, USA;
		{\tt\small (vishnuc,dpanagou)@umich.edu}}
	\thanks{This work has been funded by the Center for Unmanned Aircraft Systems (C-UAS), a National Science Foundation Industry/University Cooperative Research Center (I/UCRC) under NSF Award No. 1738714 along with significant contributions from C-UAS industry members.}
}

\begin{document}
	\maketitle
	\thispagestyle{empty}
	\pagestyle{empty}
	
	\begin{abstract}
		
		This paper studies a defense approach against one or more swarms of adversarial agents. In our earlier work, we employ a closed formation (`StringNet') of defending agents (defenders) around a swarm of adversarial agents (attackers) to confine their motion within given bounds, and guide them to a safe area. The control design relies on the assumption that the adversarial agents remain close enough to each other, i.e., within a prescribed connectivity region. To handle situations when the attackers no longer stay within such a connectivity region, but rather split into smaller swarms (clusters) to maximize the chance or impact of attack, this paper proposes an approach to learn the attacking sub-swarms and reassign defenders towards the attackers. We use a `Density-based Spatial Clustering of Application with Noise (DBSCAN)' algorithm to identify the spatially distributed swarms of the attackers. Then, the defenders are assigned to each identified swarm of attackers by solving a constrained generalized assignment problem. Simulations are provided to demonstrate the effectiveness of the approach.
		
		
		
	\end{abstract}
	
	\section{Introduction}
	
	Swarms of low-cost agents such as small aerial robots may pose risk to safety-critical infrastructure such as government facilities, airports, and military bases. Interception strategies \cite{chen2017multiplayer, coon2017control} against these threats may not be feasible or desirable in an urban environment due to posing greater risks to humans and the surrounding infrastructure. Under the assumption of risk-averse and self-interested adversarial agents (attackers) that tend to move away from the defending agents (defenders) and from other dynamic objects, herding can be used as an indirect way of guiding the attackers to a safe area.  
	
	In our recent work \cite{chipade2019swarmherding,chipade2020swarmherding}, we developed a herding algorithm, called `StringNet Herding', to herd a swarm of adversarial attackers away from a safety-critical (protected) area. A closed formation (`StringNet') of defending agents connected by string barriers is formed around a swarm of attackers staying together to confine their motion within given bounds, and guide them to a safe area. However, the assumption that the attackers will stay together in a connectivity region, and they will react to the defenders collectively as a single swarm while attacking the protected area, can be quite conservative in practice.
		
	In this paper, we build upon our earlier work on `StringNet Herding' \cite{chipade2020swarmherding} and study the problem of defending a safety-critical (protected) area from adversarial agents that \textit{may} or \textit{may not} stay together. We propose a `Multi-Swarm StringNet Herding' approach that uses clustering-based defender assignment, and the `StringNet Herding' method to herd the adversarial attackers to known safe areas.
	
	\subsubsection{Related work}
	Several approaches have been proposed to solve the problem of herding. Some examples are: the $n$-wavefront algorithm \cite{gade2015herding,paranjape2018robotic}, where the motion of the birds on the boundary of the flock is influenced based on the locations of the airport and the safe area; herding via formation control based on a potential-field approach \cite{pierson2018controlling}; biologically-inspired "wall" and "encirclement" methods that dolphins use to capture a school of fish \cite{haque2011biologically}; an RRT approach that finds a motion plan for the agents while maintaining a cage of potentials around the sheep \cite{varava2017herding}; sequential switching among the chased targets \cite{licitra2017single}.
	In general, the above approaches suffer from one or more of the following: 1) dependence on knowing the analytical modeling of the attackers' motion, 2) lack of modeling of the adversarial agents' intent to reach or attack a certain protected area, 3) simplified motion and environment models. The proposed `StringNet Herding' approach relaxes the first and the third issue above, and takes into account the second one for control design.
    Clustering of data points is a popular machine learning technique \cite{xu2015comprehensive}. 
	There are various categories of clustering algorithms: 1) partition based (K-means \cite{macqueen1967some}), 2) hierarachy based (BIRCH \cite{zhang1996birch}), 3) density based (DBSCAN \cite{ester1996density}), 4) stream based (STREAM \cite{o2002streaming}), 6) graph theory based (CLICK \cite{sharan2000click}). Spatial proximity of the agents is crucial for the problem at hand so our focus will be mostly on the density based approaches in this paper.

	Assignment problems have also been studied extensively \cite{burkard2012assignment}. In this paper, we are interested in a generalized assignment problem (GAP) \cite{oncan2007survey}, in which there are more number of objects than knapsacks to be filled. GAP is known to be NP-hard but there are approximation algorithms to solve an arbitrary instance of GAP \cite{oncan2007survey}.
	 
	\subsubsection{Overview of the proposed approach}
	
	
	The proposed approach involves: 1) identification of the clusters (swarms) of the attackers that stay together, 2) distribution and assignment of the defenders to each of the identified swarms of the attackers, 3) use of `StringNet Herding' approach by the defenders to herd each identified swarm of attackers to the closest safe area. 
	
	More specifically, we use the ``Density based Spatial Clustering of Application with Noise (DBSCAN)" algorithm \cite{ester1996density} to identify the swarms of the attackers in which the attackers stay in a close proximity of the other attackers in the same swarm. We then formulate a generalized assignment problem with additional constraints on the connectivity of the defenders to find which defender should go against which swarm of attackers and herd it to one of the safe areas. This connectivity constrained generalized assignment problem (C2GAP) is modeled as a mixed integer quadratically constrained program (MIQCP) to obtain an optimal assignment solution. We also provide a hierarchical algorithm to find the assignment quickly, which along with the MIQCP formulation is the major contribution of this paper.

	\subsubsection{Structure of the paper}
	Section \ref{sec:math_model} describes the mathematical modeling and problem statement. The StringNet herding approach is briefly discussed in Section \ref{sec:herding}. The approach on clustering and the defenders-to-attackers assignment for multiple-swarm herding is discussed in Section \ref{sec:multi_swarm_herding}. 
	Simulations and conclusions are provided in Section \ref{sec:simulations} and \ref{sec:conclusions}, respectively.
	
	\section{Modeling and Problem Statement}\label{sec:math_model}
	\textit{Notation}: The set of integers greater than 0 is denoted by $\bZ_{>0}$. Vectors and matrices are denoted by small and capital bold letters, respectively (e.g., $\textbf{r}$, $\textbf{P}$). $\norm{.}$ denotes the Euclidean norm of its argument. $\abs{.}$ denotes the absolute value of a scalar, and cardinality if the argument is a set. $n!$ is a factorial of $n$.

	We consider $N_a$ attackers $\calA_i$, $i \in I_a= \{1,2,...,N_a\}$, and $N_d$ defenders $\calD_j$, $j \in I_d= \{1,2,...,N_d\}$, operating in a 2D environment $\calW \subseteq \mathbb{R}^2$ that contains a protected area $\calP \subset \calW$, defined as $\calP=\{\textbf{r} \in \bR^2 \;| \; \norm{\textbf{r}-\textbf{r}_p}\le \rho_p\}$, and $N_s$ safe areas $\calS_{m} \subset \calW$, defined as $\calS_{m}=\{\textbf{r} \in \bR^2 \; | \; \norm{\textbf{r}-\textbf{r}_{sm}}\le \rho_{sm}\}$, for all $m \in I_s=\{1,2,...,N_s\}$, where $(\textbf r_p, \rho_p)$ and $(\textbf r_{sm}, \rho_{sm})$ are the centers and radii of the corresponding areas, respectively. The number of defenders is no less than that of attackers, i.e., $N_d \ge N_a$. The agents $\calA_i$ and $\calD_j$ are modeled as discs of radii $\rho_a$ and $\rho_d\le \rho_a$, respectively and move under double integrator (DI) dynamics with quadratic drag: 
	\begin{subequations} \label{eq:attackDyn1}
	\begin{align}
	\dot{\textbf{r}}_{ai}
	=\textbf{v}_{ai}, \quad \quad 
	\dot{\textbf{v}}_{ai}
	=\textbf{u}_{ai}-C_{D} \norm{\textbf{v}_{ai}}\textbf{v}_{ai};\\
	\dot{\textbf{r}}_{dj}
	=\textbf{v}_{dj}, \quad \quad 
	\dot{\textbf{v}}_{dj}
	=\textbf{u}_{dj}-C_{D} \norm{\textbf{v}_{dj}}\textbf{v}_{dj};\\
	\norm{\mathbf{u}_{ai}} \le \bar{u}_a, \quad
	\norm{\mathbf{u}_{dj}} \le \bar{u}_d; 
	\end{align}
	\end{subequations}
	where $C_D$ is the drag coefficient, $\mathbf{r}_{ai}=[x_{ai}\; y_{ai}]^T$ and $\mathbf{r}_{dj}=[x_{dj}\; y_{dj}]^T$ are the position vectors of $\calA_i$ and $\calD_j$, respectively; $\mathbf{v}_{ai}=[v_{x_{ai}}\; v_{y_{ai}}]^T$, $\mathbf{v}_{dj}=[v_{x_{dj}}\; v_{y_{dj}}]^T$ are the velocity vectors, respectively, and $\mathbf{u}_{ai}=[u_{x_{ai}}\; u_{y_{ai}}]^T$, $\mathbf{u}_{dj}=[u_{x_{dj}}\; u_{y_{dj}}]^T$ are the accelerations (the control inputs), respectively.	
	 The defenders are assumed to be faster than the attackers, i.e., $\bar{u}_a < \bar{u}_d$.
	This model poses a speed bound on each player with limited acceleration control, i.e., $v_{ai}=\norm{\mathbf{v}_{ai}}<\bar{v}_a=\sqrt{\frac{\bar{u}_a}{C_d}}$ and $v_{dj}=\norm{\mathbf{v}_{dj}}<\bar{v}_d=\sqrt{\frac{\bar{u}_d}{C_d}}$. 
	We assume that every defender $\calD_j$ senses the position $\textbf r_{ai}$ and velocity $\textbf{v}_{ai}$ of the attacker $\calA_i$ when $\calA_i$ is inside a circular sensing-zone $\calZ_d^s=\{\mathbf r \in \mathbb{R}^2 |\; \norm{\textbf{r}-\textbf{r}_p} \le \rho_d^s\}$. Each attacker $\calA_i$ has a similar local sensing zone $\calZ_{ai}^s=\{\textbf{r} \in \bR^2 \;| \; \norm{\textbf{r}-\textbf{r}_{ai}}\le \rho_{ai}^s \}$ inside which they sense defenders' positions and velocities.
	

	The attackers aim to reach the protected area $\calP$. The attackers may use flocking controllers \cite{dai2014flocking} to stay together, or they may choose to split into different smaller swarms \cite{goel2019leader,raghuwaiya2016formation}. The defenders aim to herd each of these attackers to one of the safe areas in $\calS=\{\calS_1,\calS_2,...,\calS_{N_s}\}$ before they reach $\calP$. We consider the following problems.
	
	\begin{problem}[Swarm Identification] \label{prob:1}
	Identify the swarms $\{\calA_{c_1},\calA_{c_2},...,\calA_{c_{N_{ac}}}\}$ of the attackers for some unknown $N_{ac}\ge 1$ such that attackers in the same swarm $\calA_{c_k}$, and only them, satisfy prescribed conditions on spatial proximity. 
	\end{problem}
	\begin{problem}[Multi-Swarm Herding]
	Find subgroups $\{\calD_{c_1},\calD_{c_2},...,\calD_{c_{N_{ac}}}\}$ of the defenders and their assignment to the attackers' swarms identified in Problem \ref{prob:1}, such that all the defenders in the same subgroup are connected via string barriers to enclose and herd the assigned attacker's swarm.
	\end{problem}

	\section{Herding a Single Swarm of Attackers}\label{sec:herding}
	To herd a swarm of attackers to $\calS$, we use `StringNet Herding', developed in\cite{chipade2020swarmherding}. StringNet is a closed net of strings formed by the defenders as shown in Fig. \ref{fig:clusterAssignment}. The strings are realized as impenetrable and extendable line barriers (e.g., spring-loaded pulley and a rope or other similar mechanism \cite{mirjan2016building}) that prevent attackers from passing through them. The extendable string barrier allows free relative motion of the two defenders connected  by the string. The string barrier can have a maximum length of $\bar{R}_{s}$. If the string barrier were to be physical one, then it can be established between two defenders $\calD_j$ and $\calD_{j'}$ only when they are close to each other and have almost same velocity, i.e., $\norm{\mathbf{r}_{dj}-\mathbf{r}_{dj'}}\le \underline{R}_s<\bar{R}_s$ and $\norm{\mathbf{v}_{dj}-\mathbf{v}_{dj'}}\le \epsilon$, where $\underline{R}_s$ and $\epsilon$ are small numbers. The underlying graph structure for the two different ``StringNet'' formations defined for a subset of defenders $\calD'=\{\calD_j \;|\; j \in I_d'\}$, where $I_d' \subseteq I_d$, are defined as follows:
	
	\begin{definition}[Closed-StringNet] \label{def:closed_StringNet} The Closed-StringNet $\calG^{s}_{cl}(I_d')= (\calV^s_{cl}(I_d'),$ $\calE^s_{cl}(I_d'))$ is a cycle graph consisting of: 1) a subset of defenders as the vertices, $\calV^s_{cl}(I_d')=\{\calD_j \;|\; j \in I_d'\}$, 2) a set of edges, $\calE^s_{cl}(I_d')=\{(\calD_j,\calD_{j'}) \in \calV^s_{cl}(I_d') \times \calV^s_{cl}(I_d') | \calD_j \overset{s} \longleftrightarrow \calD_{j'} \}$, where the operator $\overset{s} \longleftrightarrow$ denotes an impenetrable line barrier between the defenders.
	\end{definition}

\begin{definition}[Open-StringNet] The Open-StringNet $\calG^{s}_{op}(I_d')= (\calV^s_{op}(I_d'),$ $\calE^s_{op}(I_d'))$ is a path graph consisting of: 1) a set of vertices, $\calV^s_{op}(I_d')$ and 2) a set of edges, $\calE^s_{op}(I_d')$, similar to that in Definition \ref{def:closed_StringNet}.
\end{definition}
	
	The StringNet herding consists of four phases: 1) gathering, 2) seeking, 3) enclosing, and 4) herding to a safe area. These phases are discussed as follows.
	
	\subsection{Gathering}\label{sec:gatheting}
     We assume that the attackers start as single swarm that stays together, however, they may start splitting into smaller groups as they sense the defenders in their path. The aim of the defenders is to converge to an open formation $\scriptF_d^g$ centered at the gathering center $\mathbf{r}_{df^g}$ located on the
	expected path of the attackers, where the expected path is defined as the shortest path of
	the attackers to the protected area, before the attackers reach $\mathbf{r}_{df^g}$. Let $\scriptR_d(N_a):\bZ_{>0} \rightarrow \bZ_{>0}$ be the resource allocation function that outputs the number of the defenders that can be assigned to the given $N_a$ attackers. The open formation $\scriptF_d^g$ is characterized by
	the positions $\bm{\xi}_l^g$, for all 
	$l \in I_{d{c_0}}= \{1,2,...,\scriptR_d(N_a)\}$, as shown in Fig.~\ref{fig:clusterAssignment}. Once
	the defenders arrive at these positions, the defenders get
	connected by strings as follows: the defender
	at $\bm{\xi}_l^g$ gets connected to the defender at $\bm{\xi}_{l+1}^g$ for all
	$l \in \{1,2,...,\scriptR_d(N_a)-1\}$ (see Fig.~\ref{fig:clusterAssignment}). The formation $\scriptF_d^g$ is chosen to be a straight line formation as opposed to a semicircular formation\footnote{Completing a circular formation starting from a semicircular formation of the same radius is faster. However, the semicircular formation, for a given length constraint on the string barrier ($\bar{R}_s$), creates smaller blockage to the attackers as compared to the line formation. It is a trade-off between speed and effectiveness.} chosen in \cite{chipade2020swarmherding} to allow for the largest blockage in the path of the attackers. The angle made by the normal to the line joining $\bm{\xi}_1^g$ and $\bm{\xi}_{N_d}^g$ (clockwise from $\bm{\xi}_1^g$, see Fig.~\ref{fig:clusterAssignment}) is the orientation of the formation. The formation $\scriptF_d^g$ is chosen such that its orientation is toward the attackers on their expected path (defined above), see the blue formation in Fig~\ref{fig:clusterAssignment}.
	The desired positions $\bm{\xi}_l^g$ on $\scriptF_d^g$ centered at the gathering center $\mathbf{r}_{df^g}$ are:
	\be 
	\arraycolsep=1.4pt
	\baa
	\bm{\xi}_{l}^g=\mathbf{r}_{df^g} + \hat{R}_{l} \hat{\mathbf{o}} (\theta_{df^g}+\frac{\pi}{2}), \quad \text{for all } l \in I_{d{c_0}};
	\eaa
	\ee
     where $\hat{R}_{l}= \hat{R}_{d}^{d,g} \left(\frac{N_d-2l+1}{2} \right)$, $\hat{\mathbf{o}}(\theta)=[ \cos(\theta), \; \sin(\theta) ]^T$ is the unit vector making an angle $\theta$ with $x$-axis,
	 $\theta_{df^g}=\theta_{a_{cm}}^*+\pi$, where $\theta_{a_{cm}}^*$ is the angle made by the line segment joining the attackers' center of mass (ACoM) to the center of the protected area (the shortest path from the initial position of ACoM to $\calP$) with $x$-axis. These positions are static, i.e., $\dot{\bm{\xi}}_{l}^g=\ddot{\bm{\xi}}_{l}^g=\mathbf{0}$. The gathering center  $\mathbf{r}_{df^g} = \rho_{df}^g \hat{\mathbf{o}}(\theta_{df^g}) $ is such that $\rho_{df}^g> \rho_p$.
	 
	\begin{figure}[ht]
	\centering
	\includegraphics[width=.8\linewidth,trim={5.3cm 1.2cm 4.4cm 3.15cm},clip]{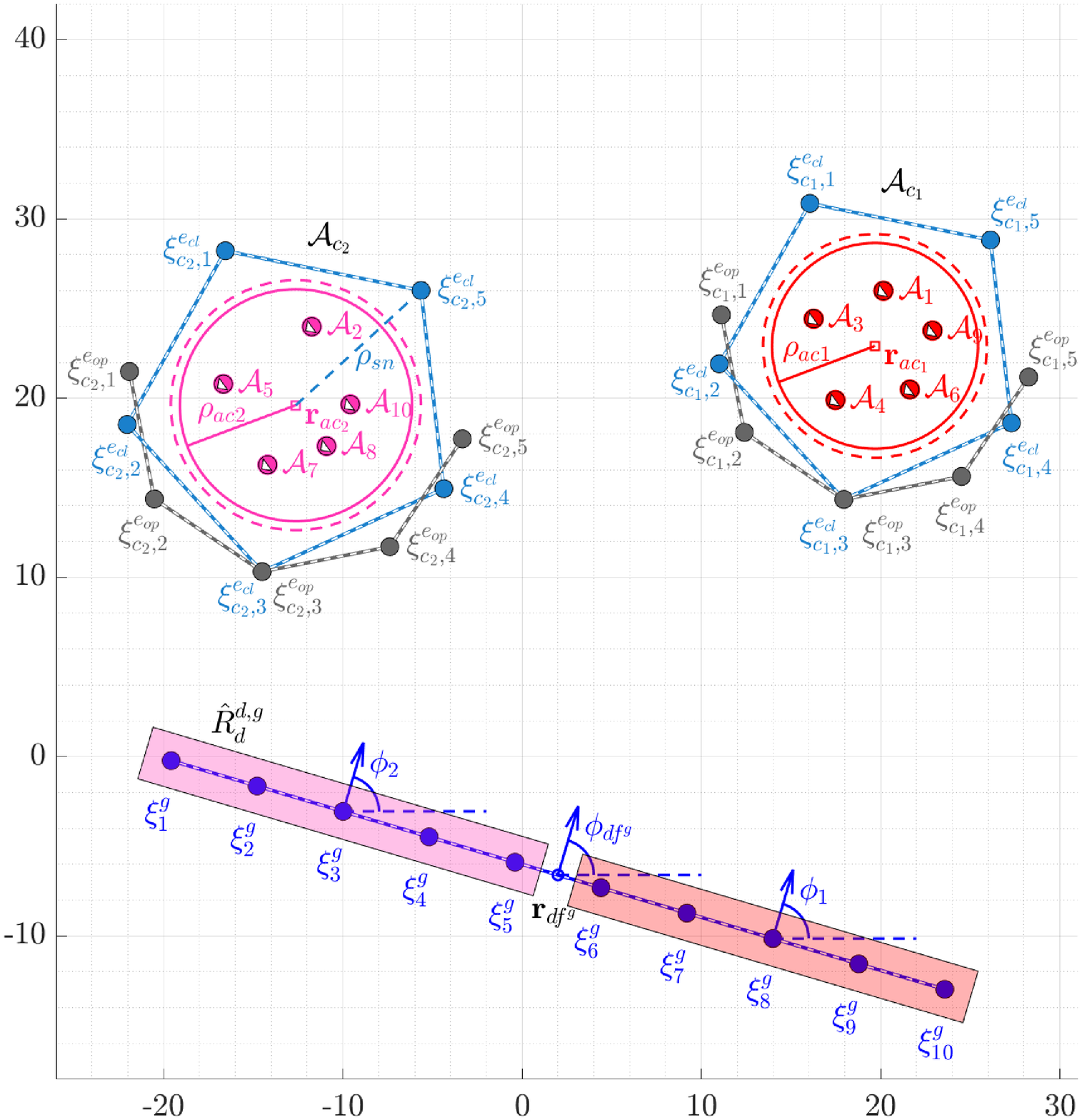}
	\caption{Assignment of defenders to the attackers' swarms}
	\label{fig:clusterAssignment}
    \end{figure}
    We define the defender-goal assignment as:
	 \begin{definition}[Defender-Goal Assignment]
	 A bijective mapping $\bm{\beta}_0:\{1,2,...,\scriptR_d(N_a)\} \rightarrow I_d$ such that the defender $\calD_{\bm{\beta}_0(l)}$ is assigned to go to the goal $\bm{\xi}_{l}^g$.
	 \end{definition}
	As discussed in \cite{chipade2020swarmherding}, we design a time-optimal motion plan so that the defenders converge to the formation $\scriptF_d^g$ as early as possible. Given initial positions for the $N_d$ defenders, and desired goal positions on the formation $\scriptF_d^g$, we recursively solve a mixed integer quadratic program (MIQP) using bisection method to find: 1) the best gathering center, if feasible, and 2) the best defender-goal assignment. 
	The MIQP finds the best defender-goal assignment by using: 1) the time information of the time-optimal trajectories obtained for each defender to go from its initial position to any goal position $\bm{\xi}_l^g$ under bounded acceleration \cite{chipade2020swarmherding}, and 2) the information of collision of all pairs of the time-optimal trajectories. The bisection method is then used to find the best gathering center by comparing the maximum time for the defenders obtained from the MIQP and the minimum time required by the attackers to reach the gathering center.

	\subsection{Seeking}
	After the defenders accomplish gathering, suppose a group of defenders $\calD_{c_k}=\{\calD_j |j \in I_{dc_k}\}$, $I_{dc_k} \subseteq I_d$, is tasked to herd a swarm of attackers $\calA_{c_k}=\{\calA_i |i \in I_{ac_k}\}$, $I_{ac_k} \subseteq I_a$, the details are discussed later in Section \ref{sec:multi_swarm_herding}. Let 
	$\bm{\beta_{k}}:\{1,2,...,|\calD_{c_k}|\}\rightarrow I_{dc_k}$ be the mapping that gives the indexing order of the defenders in $\calD_{c_k}$ on the Open-StringNet line formation $\scriptF_{dc_k}^{s}$ (similar to $\scriptF_{d}^{g}$).
	In the seeking phase, the defenders in $\calD_{c_k}$ maintain the line formation $\scriptF_{dc_k}^{s}$ and try to get closer to the swarm of attackers $\calA_{c_k}$ by using state-feedback, finite-time convergent, bounded control laws as discussed in \cite{chipade2020swarmherding}. The control actions as derived in \cite{chipade2020swarmherding} for the defenders in $\calD_{c_k}$ are modified to incorporate collision avoidance from the other StringNet formations by $\calD_{c_{k'}}$, for $k'\neq k$.
	
	\subsection{Enclosing: Closed-StringNet formation}
Once the Open-StringNet formation reaches close to the attackers' formation, the enclosing phase begins in which the defenders start enclosing the attackers by moving to their desired positions on the enclosing formations while staying connected to their neighbors. We choose two formations for this phase that the defenders sequentially achieve: 1) Semi-circular Open-StringNet formation ($\scriptF_{dc_k}^{e_{op}}$), 2) Circular Closed-StringNet formation ($\scriptF_{dc_k}^{e_{cl}}$). When the defenders directly try to converge to a circular formation from a line formation during this phase, the defenders at the either end of the Open-StringNet formation will start coming closer to each other reducing the length of the overall barrier in the attackers' path significantly. This is because the desired positions of these terminal defenders in the circular formation would be very close to each other on the opposite side of the circular formation (see Fig.~\ref{fig:clusterAssignment}) and collision avoidance part of the controller is only active locally near the circle of maximum radius $\bar{\rho}_{ac_k}$ around the swarm $\calA_{c_k}$. So the defenders would first converge to a semi-circular formation and would converge to a circular formation after the former is achieved.
	
	The desired position $\bm{\xi}_{c_k,l}^{e_{op}}$ on the Open-StringNet formation $\scriptF_{dc_k}^{e_{op}}$ (Fig.~\ref{fig:clusterAssignment}) is chosen on the circle with radius $\rho_{sn_k}$ centered at $\mathbf{r}_{\hat{ac}_k}$ as:
	\be \label{eq:goal_in_OpenStringNet}
	\arraycolsep=1.4pt
	\baa
	\bm{\xi}_{c_k,l}^{e_{op}}=\mathbf{r}_{\hat{ac}_k} + \rho_{sn_k} \hat{\mathbf{o}}(\theta_{l}) \text{, where }
	\theta_{l} = \theta_{df_k}^{e*}+\frac{\pi}{2}+\frac{\pi(l-1)}{|\calD_{c_k}|-1},
	\eaa
	\ee
	for all $l \in \{1,2,...,|\calD_{c_k}|\}$, where $\theta_{df_k}^{e*}=\theta_{df_k}^{s*}$. The center $\mathbf{r}_{\hat{ac}_k}=\mathbf{r}_{ac_k}+\tilde{\mathbf{r}}_{\hat{ac}_k}$, where $\tilde{\mathbf{r}}_{\hat{ac}_k}$ is the position of the centroid of the convex hull of the position coordinates of the attackers in $\calA_{c_k}$ relative to the center of mass $\mathbf{r}_{ac_k}=\sum_{i \in I_{ac_k}} \frac{\mathbf{r}_{ai}}{|\calA_{c_k}|}$ of $\calA_{c_k}$ at the latest time when the swarm $\calA_{c_k}$ was identified. 
	The radius $\rho_{sn_k}$ should satisfy, $\bar{\rho}_{ac_k} +b_d < \rho_{sn_k}$, where $\bar{\rho}_{ac_k}$ is maximum radius of swarm $\calA_{c_k}$. The parameter $b_d$ is the tracking error for the defenders in this phase \cite{chipade2020swarmherding}. 
	
	Similarly, the desired positions $\bm{\xi}_{c_k,l}^{e_{cl}}$ on the Closed-StringNet formation $\scriptF_{dc_k}^{e_{cl}}$ same as in Eq.~\ref{eq:goal_in_OpenStringNet} with $
	\theta_{l} = \theta_{df_k}^{e*}+\frac{\pi(2l-1)}{|\calD_{c_k}|}$,
	for all $l \in \{1,2,...,|\calD_{c_k}|\}$.
	Both the formations move with the same velocity as that of the attackers' center of mass, i.e., $\dot{\bm{\xi}}_{c_k,l}^{e_{op}}=\dot{\bm{\xi}}_{c_k,l}^{e_{cl}}=\dot{\mathbf{r}}_{ac_k}$. 
	
	The defenders $\calD_{c_k}$ first track the desired goal positions $\bm{\xi}_{c_k,l}^{e_{op}}$ by using the finite-time convergent, bounded control actions given in \cite{chipade2020swarmherding}. Once the defender $\calD_{\bm{\beta}_{k}(1)}$ and $\calD_{\bm{\beta}_{k}(|\calD_{c_k}|)}$
 reach within a distance of $b_d$ from $\bm{\xi}_{c_k,1}^{e_{op}}$ and $\bm{\xi}_{c_k,|\calD_{c_k}|}^{e_{op}}$, i.e., $\norm{\mathbf{r}_{d\bm{\beta}_{k}(1)}-\bm{\xi}_{c_k,1}^{e_{op}}}<b_d$ and $\norm{\mathbf{r}_{d\bm{\beta}_{k}(|\calD_{c_k}|)}-\bm{\xi}_{c_k,|\calD_{c_k}|}^{e_{op}}}<b_d$, respectively, the desired goal positions are changed from $\bm{\xi}_{c_k,l}^{e_{op}}$ to $\bm{\xi}_{c_k,l}^{e_{cl}}$
 for all $l \in \{1,2,...,|\calD_{c_k}|\}$. 
The StringNet is achieved when $\norm{\mathbf{r}_{d\bm{\beta}_{k}(l)}-\mathbf{\bm{\xi}}_{c_k,l}^{e_{cl}}} \le b_d$ for all $l \in \{1,2,...,|\calD_{c_k}|\}$ during this phase.
	\subsection{Herding: moving the Closed-StringNet to safe area}
	Once a group of defenders $\calD_{c_k}=\{\calD_j| j \in I_{dc_k} \}$, for $I_{dc_k}\subseteq I_d$, forms a StringNet around a swarm of attackers, they move while tracking a desired rigid closed circular formation $\mathscr{F}_{dc_k}^h$ centered at a virtual agent $\mathbf{r}_{df_k^h}$ as discussed in \cite{chipade2020swarmherding}. The swarm is herded to the closest safe area $S_{\varsigma(k)}$, where $\varsigma(k)=\ds \argmin_{m \in I_s} \norm{\mathbf{r}_{df_k^h}-\mathbf{r}_{sm }}$.

	\section{Multi-Swarm Herding}\label{sec:multi_swarm_herding}
 We consider that the attackers split into smaller groups as they sense the defenders in their path, to maximize the chance of at least some attackers reaching the protected area by circumnavigating the oncoming defenders. 
 To respond to such strategic movements of the attackers, the defenders need to collaborate intelligently. In the approach presented in this paper, the defenders first identify the spatial clusters of the attackers. Then, the defenders distribute themselves into smaller connected groups, subsets of defenders that have already established an Open-StringNet formation, in order to herd these different spatial clusters (swarms) of the attackers to safe areas. In the next subsections, we discuss the clustering and the defender to swarm assignment algorithms.

\subsection{Identifying Swarms of the Attackers} \label{sec:clustering}
In order to identify the spatially distributed clusters (swarms) of the attackers, the defenders utilize the Density Based Spatial Clustering of Applications with Noise (DBSCAN) algorithm \cite{ester1996density}. Given a set of points, DBSCAN algorithm finds clusters of high density points (points with many nearby neighbors), and marks the points as outliers if they lie alone in low-density regions (whose nearest neighbors are too far away). DBSCAN algorithm can identify clusters of any shape in the data and requires two parameters that define the density of the points in the clusters: 1) $\varepsilon_{nb}$ (radius of the neighborhood of a point), 2) $m_{pts}$ (minimum number of points in $\varepsilon_{nb}$-neighborhood of a point). In general, attackers can split into formations with varied range of densities making the choice of the parameters $\varepsilon_{nb}$ and $m_{pts}$ challenging.
Variants of the DBSCAN algorithm, such as OPTICS \cite{ankerst1999optics}, can find clusters of varying density, however, they are more time consuming. To keep computational demands low, we use the DBSCAN algorithm with fixed parameters $\varepsilon_{nb}$ and $m_{pts}$, which quickly yields useful clustering information about the attackers satisfying a specified connectivity constraints. 

The neighborhood of an attacker is defined using weighted distance between two attackers: $d(\mathbf{x}_{ai},\mathbf{x}_{ai'})=\sqrt{(\mathbf{x}_{ai}-\mathbf{x}_{ai'})^T \mathbf{M} (\mathbf{x}_{ai}-\mathbf{x}_{ai'})}$, where $\mathbf{x}_{ai}=[\mathbf{r}_{ai}^T,\mathbf{v}_{ai}^T]^T$ and $\mathbf{M}$ is a weighing matrix defined as $\mathbf{M}=diag([1,1,\varphi, \varphi])$, where $\varphi$ weights relative velocity against relative position. We choose $\varphi<1$ because relative position is more important in a spatial cluster than the velocity alignment at a given time instance. The $\varepsilon_{nb}$-neighborhood of an attacker $\calA_i$ is then defined as the set of points $\mathbf{x} \in \bR^4$ such that $d(\mathbf{x}_{ai},\mathbf{x})<\varepsilon_{nb}$.

The largest circle inscribed in the largest Closed-StrignNet formation formed by the $N_d$ defenders has radius $\bar{\rho}_{ac} =\frac{\bar{R}_s}{2} \cot(\frac{\pi}{N_d})$. Maximum radius of any cluster with $N_a$ points identified by DBSCAN algorithm with parameters $\varepsilon_{nb}$ and $m_{pts}$ is $\frac{\varepsilon_{nb}(N_a-1)}{m_{pts}-1}$. If all of the attackers were to be a single swarm enclosed inside the region with radius $\bar{\rho}_{ac}$ then we would require $\varepsilon_{nb}$ to be greater than $\frac{\bar{\rho}_{ac} (m_{pts}-1) }{N_a-1}$ in order identify them as a single cluster. So we choose $\varepsilon_{nb}=\frac{\bar{\rho}_{ac} (m_{pts}-1) }{N_a-1}$ and since we want to identify even clusters with as low as 3 agents we need to choose $m_{pts}=3$. With this parameters for DBSCAN algorithm, we have:

\begin{lemma}
Let $\{\calA_{c_1},\calA_{c_2},...,\calA_{c_{N_{ac}}}\}$ be the clusters identified by DBSCAN algorithm with $\varepsilon_{nb} =  \frac{\bar{\rho}_{ac} }{N_a-1} \floor{\frac{m_{pts}}{2}}$. For all $k \in I_{ac}=\{1,2,...,N_{ac}\}$, we have  $\rho_{ac_k}=\max_{i \in I_{ac_k}} \norm{\mathbf{r}_{ai}-\mathbf{r}_{\hat{ac}_k}} \le \frac{\bar{R}_s}{2}\cot\left(\frac{\pi}{|\calA_{c_k}|}\right)$, if $|\calA_{c_k}|>3$ and $N_a=N_d$.
\end{lemma}

As the number of attackers increases, the computational cost for DBSCAN becomes higher and looses its practical usefulness. Furthermore, the knowledge of the clusters is only required by the defenders when a swarm of attackers does not satisfy the assumed constraint on its connectivity radius. So the DBSCAN algorithm is run only for swarms of attackers $\calA_{c_k}$ for some $k \in I_{ac}$ whenever the connectivity constraint is violated by them i.e., when the radius of the swarm of attackers $\calA_{c_k}$ defined as $\rho_{ac_k}=\max_{i \in I_{ac_k}} \norm{\mathbf{r}_{ai}-\mathbf{r}_{\hat{ac}_k}}$ exceeds the value $\bar{\rho}_{ac_k}=\frac{\bar{R}_s}{2}\cot\left(\frac{\pi}{N_d}\right) \frac{|\calA_{c_k}|-1}{N_a-1}$.

	\subsection{Defender Assignment to the Swarms of Attackers}\label{sec:defender_cluster_assignment}
	As the initial swarm of attackers splits into smaller swarms, the defenders must distribute themselves into smaller groups and assign the attackers' swarms (clusters) to these groups in order to enclose these swarms and subsequently herd them to the closest safe area. Let $\calA_c=\{\calA_{c_1},\calA_{c_2},\dots, \calA_{c_{N_{ac}}}\}$ be a set of swarms of the attackers after a split event has happened at time $t_{se}$. We assume that none of the swarms in $\calA_c$ is a singular one (i.e., a swarm with less than three agents), $|\calA_{c_k}|>2$ for all $k \in I_{ac}=\{1,2,...,N_{ac}\}$. We formally define the defender to attackers' swarm assignment as: 
	 \begin{definition}[Defender-Swarm Assignment]
		A set $\bm{\beta}$ $=\{\bm{\beta}_1,\bm{\beta}_2,...\bm{\beta}_{N_{ac}} \}$ of mappings $\bm{\beta}_k: \{1,2,...,$ $\scriptR_d(|\calA_{c_k}|)\} \rightarrow I_d$, where $\bm{\beta}_k$ gives the indices of the defenders assigned to the swarm $\calA_{c_k}$ for all $k \in I_{ac}$.
	\end{definition}
	
	We consider an optimization problem to find the best defender-swarm assignment as:
	\be \label{eq:swarm_assignment_problem}
	\baa 
\bm{\beta}^\star = 	\text{argmin} & \ds \sum_{k=1}^{N_{ac}} \sum_{j'=1}^{\scriptR_d(|\calA_{c_k}|)} \norm{\mathbf{r}_{\hat{ac}_k}-\mathbf{r}_{d\bm{\beta}_k (j')}}\\
	\text{Subject to} &
	 (\calD_{\bm{\beta}_k(j')},\calD_{\bm{\beta}_k(j'-1)}) \in \calE^s_{op}(I_d), \\
	&\forall j' \in \{2, ..., \scriptR_d(|\calA_{c_k}|)\},\forall k \in I_{ac}.
	\eaa	
	\ee
	The optimization cost is the sum of distances of the defenders from the centers of the attackers' swarms to which they are assigned. This ensures that the collective effort needed by all the defenders is minimized when enclosing the swarms of the attackers. The constraints in Eq.~\eqref{eq:swarm_assignment_problem} require that all the defenders that are assigned to a particular swarm of the attackers are neighbors of each other, are already connected to each other via string barriers and the underlying graph is an Open-StringNet. Assuming $N_d=N_a$, we choose $\scriptR_d(|\calA_{c_k}|)=|\calA_{c_k}|$, i.e., the number of defenders assigned to a swarm $\calA_{c_k}$ is equal to the number of attackers in $\calA_{c_k}$. This is to ensure that there are adequate number of defenders to go after each attacker in the event the attackers in swarm $\calA_{c_k}$ disintegrate into singular swarms\footnote{In this case, herding may not be the most economical way of defense. How to handle the situations with singular swarms is out of the scope of this paper and will be studied in the future work.}. 
	
	This assignment problem is closely related to generalized assignment problem (GAP) \cite{oncan2007survey}, in which $n$ objects are to be filled in $m$ knapsacks $(n\ge m)$.  This problem is modeled as a GAP with additional constraints on the objects (defenders) that are assigned to a given knapsack (attackers' swarm). We call this constrained assignment problem as connectivity constrained generalized assignment problem (C2GAP) and provide a mixed integer quadratically constrained program	(MIQCP) to find the optimal assignment as: 
	\setlength{\abovedisplayskip}{10pt}
	\bse \label{eq:defender_swarm_assign_MIQCP}
	\begin{align}
	\text{Minimize } &  J=  \textstyle \sum_{k=1}^{N_{ac}} \sum_{j=1}^{N_d} \norm{\mathbf{r}_{\hat{ac}_k}-\mathbf{r}_{dj}}\delta_{jk} \label{eq:MIQCP_cost}\\
	\vspace{2mm}
	\text{Subject to } & \scriptstyle \sum_{k \in I_{ac}} \delta_{jk}=1, \quad \forall j \in I_d; \label{eq:MIQP_constraint_1}\\
	\vspace{2mm}
    & \scriptstyle \sum_{j \in I_{d}} \delta_{jk}=\scriptR_d(|\calA_{c_k}|), \quad \forall k \in I_{ac}; \label{eq:MIQP_constraint_2}\\
	\vspace{2mm}
	&\scriptstyle \sum_{j \in \tilde{I}_d} \delta_{jk}\delta_{(j+1)k}\ge \scriptR_d(|\calA_{c_k}|)-1, \quad \forall k \in I_{ac}; \label{eq:MIQP_constraint_3}\\
	& \scriptstyle \sum_{k \in I_{ac}} \sum_{j \in {I_d}} \delta_{jk}=\scriptR_d(N_a); \label{eq:MIQP_constraint_4}\\
	& \scriptstyle \delta_{jk}\in \{0,1\}, \quad \forall j \in I_d, k \in I_{ac};
	\end{align}
	\ese
	where $\tilde{I}_d=I_d-\{N_d\}$, $\delta_{jk}$ is a decision variable which is equal to 1 when the defender $\calD_j$ is assigned to the swarm $\calA_{c_k}$ and 0 otherwise. The constraints \eqref{eq:MIQP_constraint_1} ensure that each defender is assigned to exactly one swarm of the attackers, the  capacity constraints \eqref{eq:MIQP_constraint_2} ensure that for all $k \in I_{ac}$ swarm $\calA_{c_k}$ has exactly $\scriptR_d(|\calA_{c_k}|)$ defenders assigned to it, the quadratic constraints \eqref{eq:MIQP_constraint_3} ensure that all the defenders assigned to swarm $\calA_{c_k}$ are connected together with an underlying Open-StringNet for all $k \in I_{ac}$ and the constraint \eqref{eq:MIQP_constraint_4} ensures that all the $\scriptR_d(N_a)$ defenders are assigned to the attackers' swarms. This MIQCP can be solved using a MIP solver Gurobi \cite{gurobi}.
  As shown in an instance of the defender-swarm assignment in Fig.~\ref{fig:clusterAssignment}, the defenders at $\bm{\xi}_{l}^g$ for $l \in  \{1,2,...,5\}$ are assigned to swarm $\calA_{c_2}$ and those at $\bm{\xi}_{l}^g$ for $l \in \{6,7,...,10\}$ are assigned to swarm $\calA_{c_1}$.

\subsection{Hierarchical Approach to defender-swarm assignment}
Finding the optimal defender-swarm assignment by solving the MIQCP discussed above may not be real-time implementable for a large number of agents $(>100)$. In this section, we develop a computationally efficient hierarchical approach to find defender-swarm assignment. A large dimensional assignment problem is split into smaller, low-dimensional assignment problems that can be solved optimally and quickly. Algorithm \ref{alg:defender_cluster_assignment} provides the steps to reduce the problem of size $N_{ac}$ to smaller problems of size smaller than or equal $\underline{N}_{ac}(<N_{ac})$.
\begin{algorithm}[h]
		\caption{Defender-Swarm Assignment}
		\label{alg:defender_cluster_assignment}
		\SetKwFunction{assignHierarchical}{assignHierarchical}
		\SetKwFunction{assignMIQCP}{assignMIQCP}
		\SetKwFunction{splitEqual}{splitEqual}
		\SetKwProg{Fn}{Function}{:}{}
        \Fn{\assignHierarchical{$\scriptA, \scriptD$}}{
		\eIf{$\scriptA.N_{ac}>\underline{N}_{ac}$}{[$\scriptA^l,\scriptD^l,\scriptA^r,\scriptD^r$]=\splitEqual($\scriptA,\scriptD)$;\\
		\eIf{$\scriptA^l.N_{ac}>\underline{N}_{ac}$}{
		$\bm{\beta}^l=$\assignHierarchical($\scriptA^l, \scriptD^l$);}
		{$\bm{\beta}^l=$\assignMIQCP($\scriptA^l,\scriptD^l$);}
		\eIf{$\scriptA^r.N_{ac}>\underline{N}_{ac}$}{
		$\bm{\beta}^r=$\assignHierarchical($\scriptA^r, \scriptD^r$);}
		{$\bm{\beta}^r=$\assignMIQCP($\scriptA^r, \scriptD^r$);}
		$\bm{\beta}=\{\bm{\beta}^l,\bm{\beta}^r\};$
		}
		{$\bm{\beta}$=\assignMIQCP($\mathbf{r}_{ac}$,$\mathbf{r}_d$);}
		\Return{$\bm{\beta} =\{\bm{\beta}_1,\bm{\beta}_2,...,\bm{\beta}_{N_{ac}}\}$}
		}
		\textbf{End Function}
	\end{algorithm}
	In Algorithm \ref{alg:defender_cluster_assignment}, $\scriptA$ is a data structure that stores the information of: centers of the attackers' swarms $\mathbf{r}_{ac}=[\mathbf{r}_{\hat{ac}_1}, \mathbf{r}_{\hat{ac}_2},...,\mathbf{r}_{\hat{ac}_{N_{ac}}}]$, numbers of the attackers in each swarm $\mathbf{n}_{ac}=[|\calA_{c_1}|,|\calA_{c_2}|,...,|\calA_{c_{N_{ac}}}|]$, total number of attackers $N_a$; and $\scriptD$ is a data structure that stores the information of: defenders' positions $\mathbf{r}_{d}=\{\mathbf{r}_{dj}| j \in I_d'\}$, and the goal assignment $\bm{\beta}$. \splitEqual function splits the attackers into two groups $\scriptA^l$ and $\scriptA^r$ of roughly equal number of attackers and the defenders into two groups $\scriptD^l$ and $\scriptD^r$. The split is performed based on the angles $\psi_{k}$ made by relative vectors $\mathbf{r}_{\hat{ac}_k}-\mathbf{r}_{dc}$, for all $k \in I_{ac}$, with the vector $\mathbf{r}_{\hat{ac}_k}-\mathbf{r}_{dc}$ where $\mathbf{r}_{dc}$ is the center of $\mathbf{r}_d$. We first arrange these angles $\psi_k$ in descending order. The first few clusters in the arranged list with roughly half the total number of attackers become the left group $\scriptA^l$ and the rest become the right group $\scriptA^r$. Similarly, the left group $\scriptD^l$ is formed by the first $\scriptA^l.N_a$ defenders as per the assignment $\bm{\beta}$ and the rest defenders form the right group $\scriptD^r$. We assign the defenders in $\scriptD^l$ only to the swarms in $\scriptA^l$ and those in $\scriptD^r$ only to the swarms in $\scriptA^r$. By doing so we may or may not obtain an assignment that minimizes the cost in \eqref{eq:MIQCP_cost} but we reduce the computation time significantly and obtain a reasonably good assignment quickly. As in Algorithm \ref{alg:defender_cluster_assignment}, the process of splitting is done recursively until the number of attackers' swarms is smaller than a pre-specified number $\underline{N}_{ac}$. The function \assignMIQCP finds the defender-swarm assignment by solving \eqref{eq:defender_swarm_assign_MIQCP}. As shown in Figure~\ref{fig:runTimeAssign}, the average computation time over a number of cluster configurations and initial conditions for the hierarchical approach to assignment is significantly smaller than that of the MIQCP formulation and also the cost of the hierarchical algorithm is very close to the optimal cost (MIQCP), see Fig.~\ref{fig:assignCostError}.
	\begin{figure}[ht]
	\centering
	\includegraphics[width=.98\linewidth,trim={.6cm 0cm .4cm .75cm},clip]{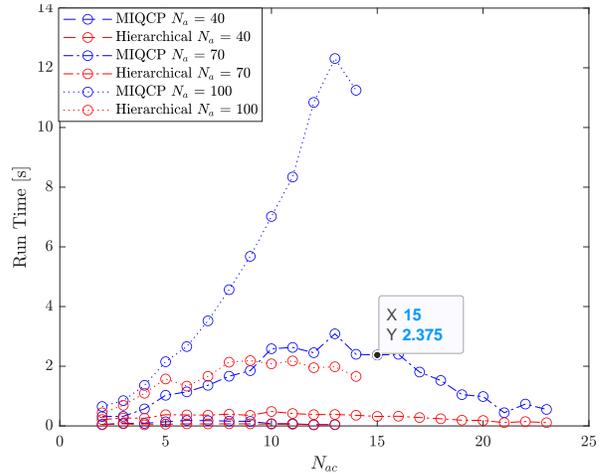}
	\caption{Run-time for assignment algorithms}
	\label{fig:runTimeAssign}
    \end{figure}
    	\begin{figure}[ht]
	\centering
	\includegraphics[width=.98\linewidth,trim={.6cm 0cm .4cm .75cm},clip]{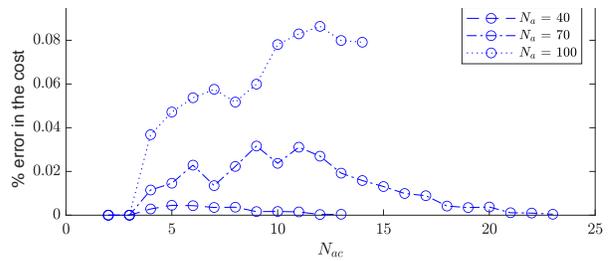}
	\caption{\% Error in the costs of the assignment algorithms}
	\label{fig:assignCostError}
    \end{figure}



	\section{Simulations}\label{sec:simulations} 
	We provide a simulation of 18 defenders herding 18 attackers to $\calS$ with bounded control inputs. Figure~\ref{fig:multiSwarmHerd} shows the snapshots of the paths taken by all agents. The positions and paths of the defenders are shown in blue color, and that of the attackers in red. The string-barriers between the defenders are shown as wide solid blue lines with white dashes in them.
	
	Snapshot 1 shows the paths during the gathering phase. As observed the defenders are able to gather at a location on the shortest path of the attackers to the protected area before the attacker reach there. Five attackers are already separated from the rest thirteen in reaction to the incoming defenders in their path. The defenders have identified two swarms of the attackers $\calA_{c_1}$ and $\calA_{c_2}$ at the end of the gathering phase and assign two subgroups $\calD_{c_1}$ and $\calD_{c_2}$ of the defenders to $\calA_{c_1}$ and $\calA_{c_2}$ using Algorithm \ref{alg:defender_cluster_assignment}. As shown in snapshot 2, $\calD_{c_1}$ and $\calD_{c_2}$ seek $\calA_{c_1}$ and $\calA_{c_2}$, but the attackers in swarm $\calA_{c_2}$ further start splitting and the defenders identify this newly formed $\calA_{c_2}$ and $\calA_{c_3}$ at time $t=120.11 s$. The group $\calD_{c_2}$ is then split into two subgroups $\calD_{c_2}$ and $\calD_{c_3}$ of appropriate sizes and assigned to the new swarms $\calA_{c_2}$ and $\calA_{c_3}$ using Algorithm \ref{alg:defender_cluster_assignment}. 
	
	Snapshot 3 shows how the 3 subgroups of the defenders are able to enclose the the identified 3 swarms of the attackers by forming Closed-StringNets around them. Snapshot 4 shows how all the three enclosed swarms of the attackers are taken to the respective closest safe areas while each defenders' group ensures collision avoidance from other defenders' groups. Additional simulations can be found at \href{https://drive.google.com/drive/folders/11qJxjlxR_AWbc4vicIRchZpvCaByCJGP?usp=sharing}{/drive/video}.
	\begin{figure*}[h]
		\centering
		\begin{subfigure}[h]{0.48\linewidth}
		\includegraphics[width=1\linewidth,trim={.1cm 0.2cm 2.cm .3cm},clip]{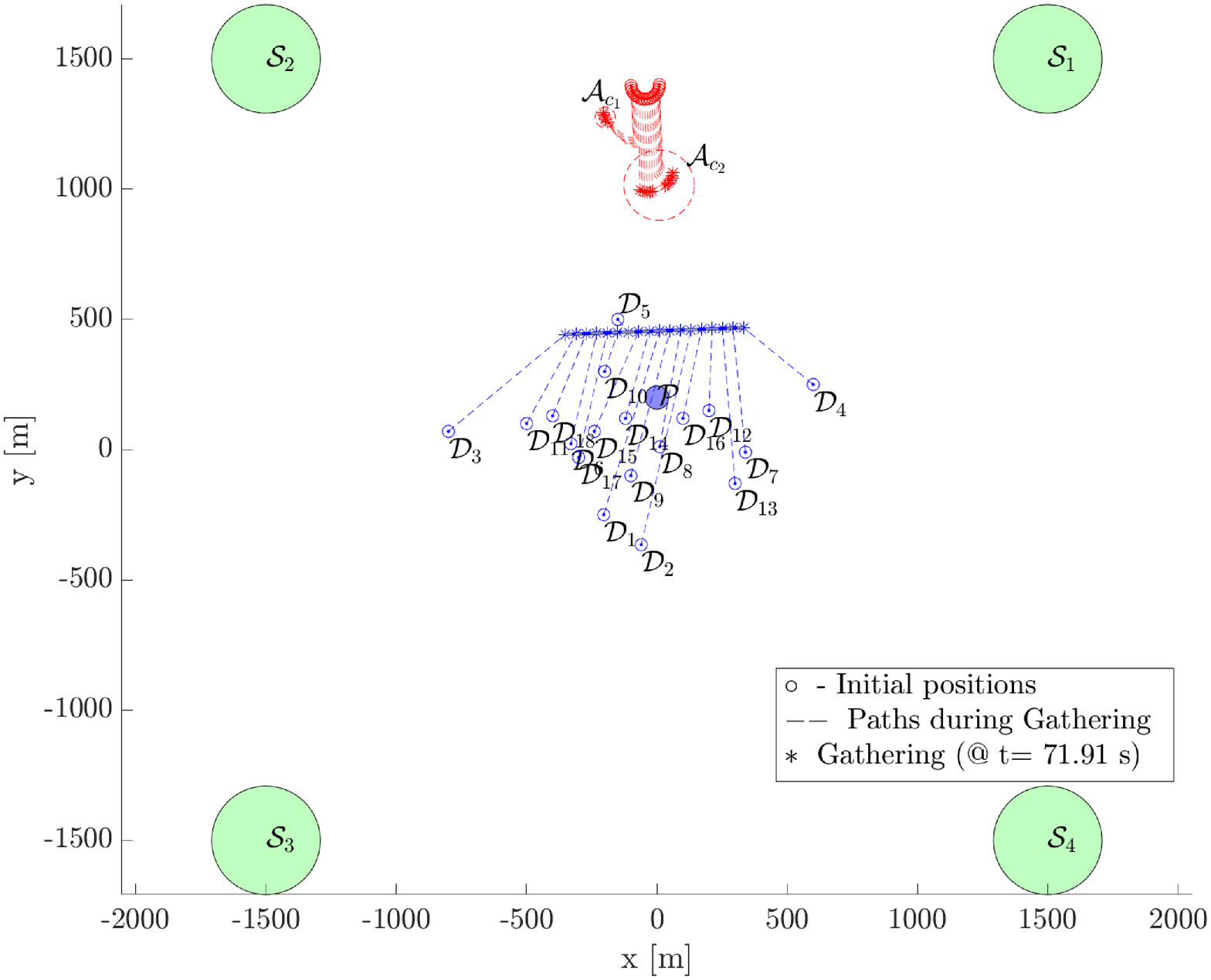}
			\label{fig:multiSwarmHerd1}
		\end{subfigure}	
		\begin{subfigure}[h]{0.48\textwidth}
		\includegraphics[width=1\linewidth,trim={.9cm 0.2cm 2.cm .65cm},clip]{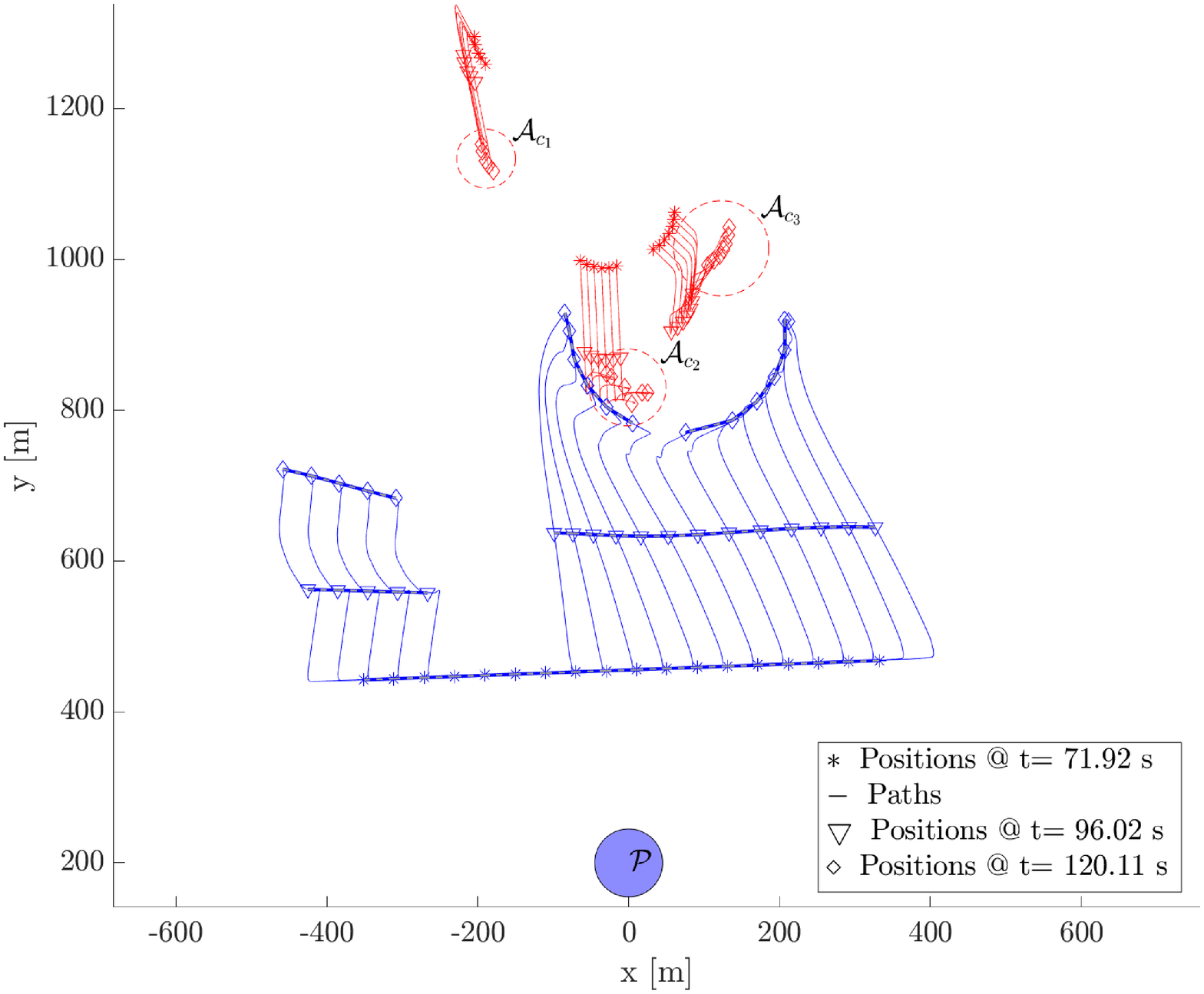}
		\label{fig:multiSwarmHerd2}
		\end{subfigure}
		\begin{subfigure}[h]{0.48\textwidth}
		\includegraphics[width=1\linewidth,trim={.9cm 0.2cm 2cm .85cm},clip]{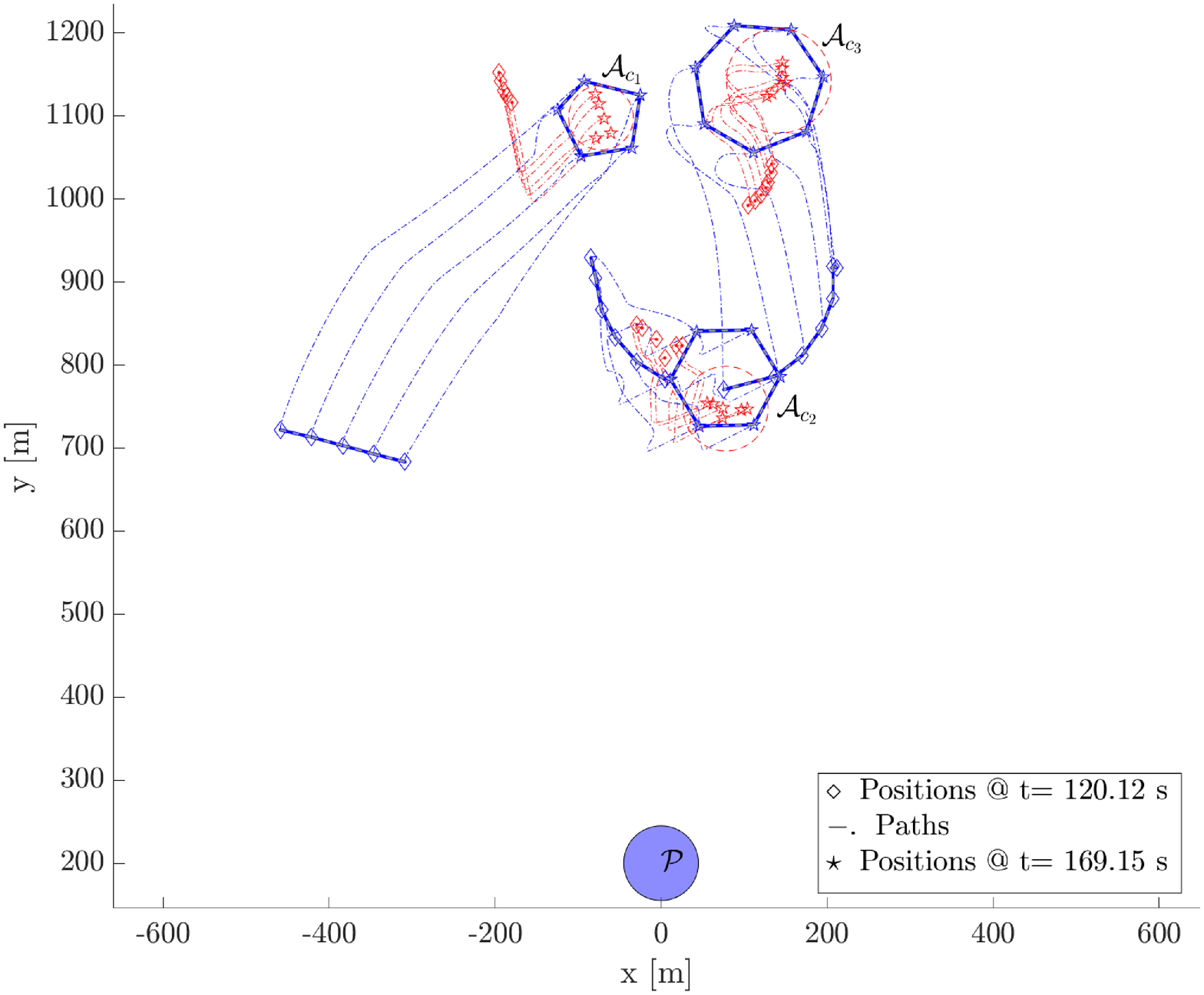}
		\label{fig:multiSwarmHerd3}
		\end{subfigure}
		\begin{subfigure}[h]{0.48\textwidth}
		\includegraphics[width=1\linewidth,trim={.9cm 0.2cm 2cm 1.25cm},clip]{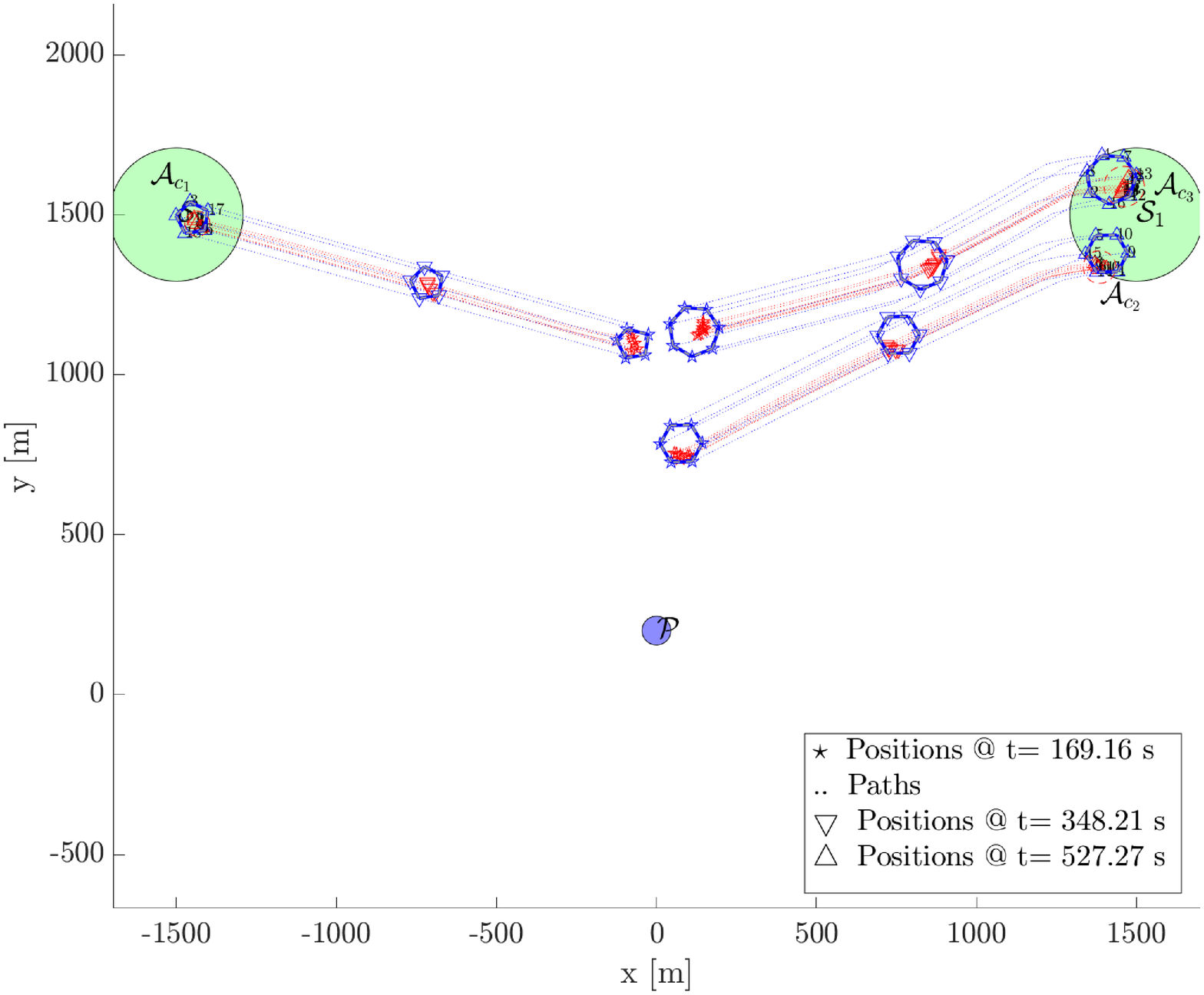}
		\label{fig:multiSwarmHerd4}
		\end{subfigure}
		\caption{Snapshots of the paths of the agents during Multi-Swarm StringNet Herding}
		\label{fig:multiSwarmHerd}
	\end{figure*}
	
	\vspace{-1mm}
	\section{Conclusions} \label{sec:conclusions}
	We proposed a clustering-based, connectivity-constrained assignment algorithm that distributes and assigns groups of defenders against swarms of the attackers, to herd them to the closest safe area using `StringNet Herding' approach. We also provide a heuristic for the defender-swarm assignment based on the optimal MIQCP that finds the assignment quickly. Simulations show how this proposed method improves the original 'StringNet Herding' method and enables the defenders herd all the attackers to safe areas even though the attackers start splitting into smaller swarms in reaction to the defenders.

	\vspace{-2mm}
	\bibliographystyle{IEEEtran}
	\bibliography{CDC2020_Refs}
\end{document}